\documentclass[letterpaper, 10pt, conference]{ieeeconf}

\usepackage{amssymb,amsmath,amsfonts}
\usepackage[]{algorithm2e}
\usepackage{blkarray}
\usepackage{bbm}
\usepackage{bm}
\usepackage{cases}
\usepackage{cite}
\usepackage[usenames,dvipsnames]{color}
\usepackage{dsfont}
\usepackage{float}
\usepackage{flushend}
\usepackage{graphicx}
\usepackage{hyperref} 
\usepackage[latin9]{inputenc}
\usepackage{mathtools}
\usepackage{mathrsfs}
\usepackage{setspace}
\usepackage{soul} 
\usepackage{subfigure}
\usepackage{stfloats}
\usepackage{tikz}
\usepackage{url}
\usepackage{verbatim}
\usepackage{xspace}
\usepackage{algpseudocode}

\usepackage{caption}

\newtheorem{ass}{Assumption}
\newtheorem{theorem}{Theorem}
\newtheorem{defn}{Definition}
\newtheorem{rem}{Remark}

\newtheorem{cor}{Corollary}

\def\ve{\varepsilon}

\def\mb{\mathbf}

\def\mc{\mathcal}

\def\mean#1{\left< #1 \right>}

\IEEEoverridecommandlockouts                
\begin{document}
\title{\bf
Distributed Constraint-Coupled Optimization over Lossy Networks }
\author{Mohammadreza Doostmohammadian, Usman A. Khan, Alireza Aghasi,
and Themistoklis Charalambous
\thanks{M. Doostmohammadian is with the Department of Electrical Engineering and Automation, School of Electrical Engineering, Aalto University, Finland, Email: \texttt{name.surname@aalto.fi} and also with the Faculty of Mechanical Engineering at Semnan University,  Iran, email: \texttt{doost@semnan.ac.ir}. 
U. A. Khan is with the Department of Electrical and Computer Engineering, Tufts University, MA, USA, email: \texttt{khan@ece.tufts.edu}. 
A. Aghasi is with the Electrical Engineering and Computer Science Department, Oregon State University, USA, email: \texttt{ali.aghasi@gmail.com}.
T. Charalambous is with the Department of Electrical and Computer Engineering, School of Engineering, University of Cyprus, Nicosia, Cyprus, email: \texttt{surname.name@ucy.ac.cy}.  } 
\thanks{This work is supported in part by the European Commission through the H2020 Project FinEst Twins under Agreement 856602.}}

\maketitle

\begin{abstract}
	This paper considers distributed resource allocation and sum-preserving constrained optimization over lossy networks, where the links are unreliable and subject to packet drops. We define the conditions to ensure convergence under packet drops and link removal by focusing on two main properties of our allocation algorithm: (i) The weight-stochastic condition in typical consensus schemes is reduced to balanced weights, with no need for readjusting the weights to satisfy stochasticity. (ii) The algorithm does not require all-time connectivity but instead uniform connectivity over some non-overlapping finite time intervals. First, we prove that our algorithm provides primal-feasible allocation at every iteration step and converges under the conditions (i)-(ii) and some other mild conditions on the nonlinear iterative dynamics. These nonlinearities address possible practical constraints in real applications due to, for example, saturation or quantization among others.  
	Then, using (i)-(ii) and the notion of bond-percolation theory, we relate the packet drop rate and the network percolation threshold to the (finite) number of iterations ensuring uniform connectivity and, thus, convergence towards the optimum value.              
	
\end{abstract}

\section{Introduction} \label{sec_intro}
Distributed algorithms are widely considered for optimization and learning over networks with several seminal works addressing the case of networks with reliable links \cite{falsone2020tracking,wu2021new,Carli_ADMM,wang2020dual,9683726,rikos2021optimal,boyd2006optimal,cherukuri2015distributed,mrd_2020fast,mrd_vtc}. Due to the unreliable nature of wireless communication, these results are not readily applicable in wireless networks. In such networks, randomness of the links should be considered as link-failures due to packet dropouts are inevitable. The unreliability of information exchange in such settings can influence the algorithm's convergence and solution feasibility.

Equality-constraint optimization problems intend to optimize the allocation cost \textit{subject to constant sum of overall resources}. The existing literature either focus on dual-based methods and ADMM solutions \cite{falsone2020tracking,wu2021new,Carli_ADMM,wang2020dual,9683726} or primal-feasible gradient-tracking methods \cite{rikos2021optimal,boyd2006optimal,cherukuri2015distributed,mrd_2020fast,mrd_vtc}. These allocation algorithms are developed over reliable communication networks and require either (i) stochastic weight design over the network or (ii) all-time network connectivity as a giant component \cite{boyd2006optimal,cherukuri2015distributed} or both \cite{wang2020dual,Carli_ADMM,falsone2020tracking,rikos2021optimal,wu2021new,9683726}. However, in unreliable networks with packet loss and message drops over the communication links, these restrictive conditions do not necessarily hold. As the messages are not delivered over some links at different times, the network resembles a dynamic graph topology over time that may even lose its connectivity at some iterations, failing condition (ii) for high drop rates. On the other hand, even assuming a connectivity guarantee, condition (i) mandates robust consensus algorithms (as in \cite{6426252,cons_drop_siam}) to modify the stochastic weights on the shared information after the drop of some messages.

The proposed allocation algorithm in this paper relaxes the stochastic condition in \cite{falsone2020tracking,rikos2021optimal,wu2021new,Carli_ADMM,wang2020dual,9683726} on the adjacency matrix to being only a \textit{balanced} weight matrix. Assuming the common knowledge of the delivered messages over the undirected links (i.e., both sender and receiver know whether its packet is delivered) \cite{olfati_rev,6426252}, there is no need for compensation strategies to readjust the weights to make them stochastic again in case of link removal or packet drops. Further, our solution only needs uniform-connectivity over time relaxing the all-time connectivity requirement over undirected  \cite{wang2020dual,Carli_ADMM,falsone2020tracking,wu2021new,9683726,boyd2006optimal} and directed networks \cite{cherukuri2015distributed,rikos2021optimal}. This implies that the network can lose path-connectivity (i.e., connectivity over a sequence of directly-connected nodes) between some nodes at some iterations. 

In this paper, the loss of packets at every iteration is modeled as removing the associated links over the network. Therefore, for high packet drop rates, a connected network may go through a transition phase and lose its connectivity (due to link failure). This is studied via percolation theory in network science literature \cite{wierman2010percolation} and is discussed in Section~\ref{sec_back} in details. 
In particular, \textit{bond-percolation} refers to phase-transition (or percolation) in the network connectivity under certain rate of link failure or removal \cite{mohseni2021percolation}. Similar threshold-based approaches are adopted in the epidemic processes with a spread of a certain disease or virus over the network \cite{li2013epidemic,van2010graph}. As a typical approach, we analyze this bond-percolation analysis over general random networks, where the randomness stems from the link failures, e.g., due to random packet dropouts in wireless sensor networks \cite{jakovetic2013distributed}. Some well-known existing random models include Erdos-Renyi (ER), Scale-Free (SF), and Small-World (SW) networks. Such models are further known to be a \textit{typical} representative of networks encountered in many real-world applications. Therefore, the literature focuses on studying such networks' properties to understand and resemble the behaviour of a diverse range of practical large-scale networks, including IoT, social, financial, and transportation networks. We relate the packet drop rate to the uniform connectivity of the network and, in turn, its bond-percolation threshold such that the convergence of the optimization algorithm is guaranteed. 

\textit{Paper Organization:} Section~\ref{sec_setup} formulates the problem with some background on the bond-percolation theory. Section~\ref{sec_alloc} provides the proposed allocation algorithm and its convergence analysis. Section~\ref{sec_drop} states the convergence under packet drops. Section~\ref{sec_sim} and \ref{sec_con} provide simulations and concluding remarks. 
 
\section{Problem Setup} \label{sec_setup}
\subsection{Background on Graph Theory} \label{sec_back}
In many networked applications, the topology and interactions of entities resemble a graph model. In such a network topology, connectivity plays a vital role in the convergence of the adopted algorithms. Any change in the network may cause a transition from connected to disconnected, which is studied via percolation theory. 

\begin{defn} \cite{wierman2010percolation}
Given the probability of link removal equal to $p$, \textit{bond-percolation} is defined as the probability threshold $p_c$ such that for $p>p_c$ there is no giant connected component in the network with probability $1$, i.e., the network loses its connectivity.
\end{defn}
Recall that, following Kolmogorov's zero-one law, one can also claim that for probability of link removal $p<p_c$ the network preserves its connectivity \cite{wierman2010percolation}. In other words, this \textit{critical} probability $p_c$ represents a two-sided \textit{phase transition} point in terms of network connectivity.

In the perspective of packet drop analysis, this percolation threshold is similarly defined as the probability where there exists \textit{no reliable delivery-path} between (at least) two nodes over the network, where a path between nodes $i,j$ denotes a sequence of linked nodes starting at $i$ and ending at $j$. Due to its complexity, no analytical solution exists to define the bond-percolation threshold for \textit{general} networks, and the existing literature is mostly limited to lattice and grid networks. For example, some rigorous bounds for different lattice graphs are given in \cite{li2021percolation,wierman2010percolation}. However, many experimental and numerical works on this problem exist in the literature based on the Monte-Carlo simulation \cite{mohseni2021percolation}. 
For a survey of percolation theory over complex networks and wireless networks see, for example, \cite{li2021percolation} and \cite{review_perc}. 

On the other hand, random graph models are analyzed to estimate the percolation thresholds of similar large-scale and complex real networks. 
An interesting result is reported following the nearest neighbour rule and the well-known \textit{scale-free} (SF) model, assuming a set of nodes scattered via the Poisson point process, for example, in networks of cellular or mobile phone base stations. One can show that connectivity of nodes to (at least) $m=3$ nearest neighbours, in the presence of lossy links and drop-outs, guarantees that the entire network remains connected \cite{review_perc}. For the ER random graph models \cite{erdos1960evolution}, the bond-percolation threshold is defined equal to $\frac{1}{\mean{\mc{N}}}$ with $\mean{\mc{N}}$ denoting the average node degree. For the Small-World (SW) networks, this threshold is exactly determined in \cite{moore2000exact}. Table~\ref{tab_perc} summarizes the bond-percolation thresholds for some typical networks.  

\begin{table} [h] 
		\centering
		\caption{Bond-percolation threshold for different networks }
		\label{tab_perc}
		\begin{tabular}{|c|c|}
			\hline
		   Network type & Bond-percolation threshold $p_c$ \\
		   	\hline
			 Square Grid \cite{mohseni2021percolation} &  $0.5$ \\
			\hline
			 ER \cite{erdos1960evolution} &  $\frac{1}{\mean{\mc{N}}}$ with $\mean{\mc{N}}$ as  mean degree  \\
			\hline
			 SW ($m=1$) \cite{moore2000exact} &  $\frac{-2\theta -1 + \sqrt{4\theta^2 + 12 \theta +1} }{4\theta}$ \\ & with $\theta$ as  shortcut probability \\
			 \hline
			 SF \cite{li2021percolation} &  $\frac{\zeta(\sigma -1,\mc{N}_m)}{\zeta(\sigma -2,\mc{N}_m) - \zeta(\sigma - 1,\mc{N}_m)}$ \\ & with $\zeta(\cdot, \cdot)$ as Hurwitz zeta function, \\ & $\mc{N}_m$ as  min node degree, \\ & and $\sigma$ as  power-law degree\\
			\hline
			\hline
		\end{tabular}
\end{table}

Next, we provide some relevant notions on algebraic graph theory \cite{diestel2017graph,graph_handbook} and Laplacian analysis over graphs \cite{SensNets:Olfati04}. Denote by $W$ the adjacency matrix associated with the network $\mc{G} = \{\mc{V},\mc{E}\}$, with weights $W_{ij} > 0$ as the weight on the link $(j, i) \in \mc{E}$ and $0$ for  $(j, i) \notin \mc{E}$. If the network is undirected with balanced weights, $W$ is symmetric. For such a network, define $\mc{N}_i = \{j | (j,i) \in \mc{N}_i\}$ as the neighborhood of node $i$. Define the diagonal matrix $\mc{D} = \mbox{diag}(\sum_{i=1}^n W_{ij})$ and the Laplacian matrix as $L = \mc{D} - W$. For a connected network (containing a spanning-tree) with symmetric $L$, its eigen-spectrum includes one and only one $0$ eigenvalue associated with the eigen-vector $\mb{1}_n$. Recall that spectral localization of the $L$ matrix plays a key role on the convergence properties of the consensus algorithms and distributed optimization methods, see more details in \cite{SensNets:Olfati04}. 

\subsection{The Equality-Constraint Allocation Problem}
The constraint-coupled optimization problem in this work is defined as follows:
\begin{align} \label{eq_dra}
	\mc{P}_1:~~	\min_\mb{x}
		~~ & F(\mb{x}) = \sum_{i=1}^{n} f_i(x_i),~
		\text{s.t.} ~  \sum_{i=1}^n x_i  = b,
\end{align}
Intuitively, solving $\mc{P}_1$ gives the optimal allocation of resources $\mb{x}$ for which the cost $F(\mb{x})$ is minimized. Recall that the equality-constraint $\sum_{i=1}^n x_i  = b$ (or $\mb{1}_n^\top \mb{x} = b$) is known as the \textit{feasibility constraint} or the \textit{sum-preserving constraint}.  

\begin{ass} \label{ass_lips_strict}
	The local objectives $f_i(x_i):\mathbb{R} \rightarrow \mathbb{R}$, ${i \in \{1,\dots,n\}}$ are strictly convex, differentiable with locally Lipschitz derivatives such that $\partial^2 f_i(x_i) < 2 u$.
\end{ass}

In general, there might be some local box constraints  $m_i \leq x_i \leq M_i$ involved in $\mc{P}_1$ that can be addressed using additive penalty terms \cite{nesterov1998introductory,bertsekas1975necessary} or the so-called barrier functions \cite{wu2021new}. Proper initialization for $x_i(0)$ under such box constraints is discussed in \cite{cherukuri2015distributed}. 
Following the KKT conditions and under Assumption~\ref{ass_lips_strict}, it is clear that for the unique optimal point $\mb{x}^*$ we have $\nabla F(\mb{x}^*) = \varphi^* \mb{1}_n$ with $\mb{1}_n$ as the column vector of 1s and $\nabla F(\mb{x}^*) := (\partial f_1(x^*_1); \cdots; \partial f_n(x^*_n))$, assuming that this $\mb{x}^*$ is in the range of the box constraints. 

\section{The Allocation Algorithm} \label{sec_alloc}
We consider the following gradient-tracking dynamics to solve $\mc{P}_1$:

\small \begin{align} \nonumber
	x_i&(k+1) = x_i(k) \\
	&-\eta \sum_{j \in \mc{N}_i} W_{ij} g_n\Big(g_l(\partial f_i (x_i(k)) - g_l(\partial f_j (x_j(k))\Big) 
	\label{eq_sol}
\end{align} \normalsize
with $g_n(\cdot)$ and $g_l(\cdot)$ representing some nonlinear functions on the nodes and the links. The following assumptions hold throughout the paper.

\begin{ass} \label{ass_con}
    All the links over the network $\mc{G}(k)$ are undirected and weight-balanced, i.e., the adjacency matrix $W(k)$ is symmetric. Further, there exists a  $B \in \mathbb{Z}_{\geq 0}$ such that the network $\mc{G}_B(k) = \bigcup_{k}^{k+B} \mc{G}(k)$ is connected, i.e., the union network $\mc{G}_B(k)$ includes a spanning-tree for all $k\geq 0$ (implying uniform connectivity or B-connectivity).
\end{ass}

Given the network structure, its associated balanced weight matrix $W$ can be designed in a distributed manner using the strategy in \cite{2014:ISCCSP2}. In case of quantized updates, similar to allocation in \cite{magnusson2018communication,rikos2021optimal}, \textit{integer} weight-balancing strategy \cite{rikos2020distributed} can be used, however, with convergence to the $\ve$-neighborhood of the optimizer \cite{magnusson2018communication}. 

\begin{ass} \label{ass_gx}
The nonlinear mappings $g_n,g_l:\mathbb{R} \rightarrow \mathbb{R}$ are odd and ``strongly'' sign-preserving with $g(z)z>0$ for $z\neq 0$, $g(0)=0$, and $\lim_{z\rightarrow 0}\frac{g(z)}{z} \neq 0$. Further, 
\begin{align} \label{eq_boundkK}
    \kappa \leq \frac{g(z)}{z} \leq \mc{K}
\end{align}
with $\kappa_l,\mc{K}_l$ and $\kappa_n,\mc{K}_n$ defined for $g_l(\cdot)$ and $g_n(\cdot)$, respectively.
\end{ass}

Such $g_n(\cdot)$ and $g_l(\cdot)$ include, for example, all monotonically increasing and Lipschitz functions. It can be proved that under Assumptions~\ref{ass_con}-\ref{ass_gx}, the gradient-tracking solution \eqref{eq_sol} satisfies the feasibility constraint at all times $k$. This is called \textit{all-time feasibility}, and, in resource allocation perspective, implies that the assigned resources $\sum_{i=1}^n x_i$ always meet the demand $b$ and is a privilege of our gradient-tracking solution \eqref{eq_sol}. Note that violating this constraint can damage or disrupt the service in many applications \cite{wu2021new}. Our proposed allocation scheme is summarized in Algorithm~\ref{alg_ac}.
\begin{algorithm}
 \KwData{Input:  $W$, $\mc{N}_i$, $\eta$, $f_i(\cdot)$}
 \KwResult{Output: Final state $\mb{x}(k)$ and cost $F(\mb{x}(k))$}
 {\textbf{Initialization:} Every node $i$ sets $k=0$ and randomly chooses $m_i \leq x_i(0) \leq M_i$ satisfying the feasibility}
 \While{termination criteria NOT true\;
 }{Each node $i$ receives $g_l(\partial f_j (x_j(k))$ from $j\in \mc{N}_i$ \;
 Updates $x_i(k+1)$ via Eq.~\eqref{eq_sol} \;
 Shares $\partial f_i (x_i(k+1))$ with its neighboring nodes $j \in \mc{N}_i$  \;
 Sets $k \leftarrow k+1$  \;
 }
\caption{\textsf{The Resource Allocation Algorithm}}
\end{algorithm}

Example applications in economic dispatch problem \cite{mrd_2020fast} and CPU scheduling and battery reservation over the smart grid \cite{mrd_vtc} are given for $g_n(z) = z|z|^{v_1-1} + z|z|^{v_2-1}$ with $0<v_1<1, v_2>1$, and logarithmic quantizer $ q_l(z) = \mbox{sgn}(z) \exp(q_{u}(\log(|z|,\rho)))$
with $q_{u}(z) = \rho \left[ \frac{z}{\rho}\right]$ as the uniform quantizer, $\mbox{sgn}(\cdot)$ as the sign function, and $[\cdot]$ as rounding to the nearest integer. In case the function $g_n(\cdot)$ or $g_l(\cdot)$ is sign-preserving, but violates condition \eqref{eq_boundkK}, one can guarantee convergence to the $\varepsilon$-neighborhood of the optimizer $\mb{x}^*$; for example, with uniform quantization \cite{mrd20211st} and single-bit data exchange scenarios \cite{taes2020finite}.   

\begin{theorem} \label{thm_conv}
    Let Assumptions~\ref{ass_lips_strict}-\ref{ass_gx} hold. The dynamics converges to the optimal solution of $\mc{P}_1$ for $\eta (B+1)  <  \overline{\eta} := \frac{\kappa_n \kappa_l \lambda_2}{u \lambda_n^2 \mc{K}_n^2 \mc{K}_l^2}$ 
    with $\lambda_2,\lambda_n$ as the smallest non-zero and largest eigenvalue of $\mc{G}_B(k)$ for all $k\geq 0$.  
\end{theorem}

\begin{proof}
   We provide the sketch of the proof here.
   
   \textbf{Uniqueness:} Following Assumption~\ref{ass_lips_strict} and the KKT conditions one can prove that there exists a unique optimizer $\mb{x}^*$ to $\mc{P}_1$ satisfying $\nabla F \in \mbox{span}(\mb{1}_n)$ \cite{boyd2006optimal}. 
   
   \textbf{Feasibility:}  Under Assumptions~\ref{ass_con}-\ref{ass_gx}, recall from \cite[Lemma~3]{mrd_2020fast} that, for $\mb{z} \in \mathbb{R}^n$,
   	\begin{align} \nonumber
	\sum_{i=1}^n z_i\sum_{j =1}^n W_{ij} g_n(g_l(z_j) -g_l(z_i)) &= \\
	- \frac{1}{2} \sum_{i,j=1}^n W_{ij}(z_j-z_i) &g_n(g_l(z_j)-g_l(z_i)).
	\end{align}
   Substituting for $\mb{z} = \nabla F(\mb{x}(k))$, we have $\mb{1}_n^\top\mb{x}(k+1) = \mb{1}_n^\top\mb{x}(k)$ and  the proof follows.
   
   \textbf{Convergence:}
   Define $\overline{F}(k) := F(\mb{x}(k))-F(\mb{x}^*)$, ${\delta \mb{x}(k) := \mb{x}(k+B)-\mb{x}(k)}$. Following the strong-convexity in Assumption~\ref{ass_lips_strict} we have \cite{nesterov1998introductory} 
   \begin{align} \nonumber
	F(\mb{x}(k+B)) &\leq F(\mb{x}(k)) \\
	&+ \nabla F(\mb{x}(k))^\top \delta \mb{x}(k) +  u\delta \mb{x}(k)^\top \delta \mb{x}(k)
	\label{eq_taylor_2}
    \end{align}
    We first consider the case $B=1$.
    To prove $\overline{F}(k+1) \leq \overline{F}(k)$, we need to show 
\begin{align} \label{eq_proof1}
\nabla F^\top \delta \mb{x}  + u \delta \mb{x}^\top \delta \mb{x}  \leq 0
\end{align} 
where we dropped the $k$ for notation simplicity. Recall the definition of the Laplacian-gradient tracking dynamics in \cite{cherukuri2015distributed} and define dispersion parameter $\xi(k) := \nabla F(k) - \mb{1}_n^\top \nabla F(k)$. Using the results in \cite[Section~V]{SensNets:Olfati04}, from \eqref{eq_sol} and some mathematical manipulations based on Assumptions~\ref{ass_con}-\ref{ass_gx} it is sufficient that
\begin{align} 
\label{eq_proof_rho}
(-\kappa_n \kappa_l \eta \lambda_2 + u \lambda_n^2 \mc{K}_n^2 \mc{K}_l^2 \eta^2)    \xi^\top    \xi \leq 0 
\end{align}
with the strict inequality for
\begin{align} \label{eq_eta}
\eta <  \frac{\kappa_n \kappa_l \lambda_2}{u \lambda_n^2 \mc{K}_n^2 \mc{K}_l^2}
\end{align}
and for $ \xi=\mb{0}_n$ holds the equality implying $\nabla F \in \mbox{span}(\mb{1}_n)$. For $B\geq 1$, the right-hand-side of \eqref{eq_eta} changes to $\eta (B+1)$ with parameters $\lambda_2, \lambda_n$ defined for $\mc{G}_B$.
This completes the proof.
\end{proof}

Note that, in the case of versatile network topology, one can choose the ratio $\frac{\lambda_2}{\lambda_n^2}$ for the spanning tree contained by $\mc{G}_B(k)$ as the minimum connectivity ensured by Assumption~\ref{ass_con}. This min spanning tree problem could be a promising direction of future research to determine the bound on the convergence rate of the algorithm. 

Some relevant results on the eigen-spectrum of $L$ are given in the literature, for example, to estimate the spectral range based on the node degrees \cite{SensNets:Olfati04}, its variation under link removal/addition \cite{van2010graph}, bounds on the algebraic connectivity $\lambda_2(\mc{G}) \geq \frac{1}{nd_g}$ with $d_g$ as the network diameter \cite[p. 571]{graph_handbook}.   

The following remark distinguishes this work in terms of packet-drop tolerance. 

\begin{rem} In the proposed allocation dynamics \eqref{eq_sol}:
\begin{enumerate}
    \item In contrast to \cite{boyd2006optimal,falsone2020tracking,rikos2021optimal,wu2021new,Carli_ADMM,wang2020dual,9683726}, we do not require the weight matrix to be stochastic, but only to be symmetric. Therefore, there is no need for weight compensation strategies  \cite{6426252,cons_drop_siam} after packet drops or any change in the network.
    \item The algorithm converges under uniform connectivity, in contrast to all-time connectivity requirement, e.g., in \cite{boyd2006optimal,cherukuri2015distributed,rikos2021optimal} or ADMM-based solutions \cite{wang2020dual,Carli_ADMM,falsone2020tracking,9683726,wu2021new}. Therefore, although the network may lose its connectivity over some time periods, due to a high rate of packet drops or switching over sparse topologies, uniform connectivity over every $B \in \mathbb{Z}_{\geq 0}$ steps  is sufficient for convergence. 
\end{enumerate}
\end{rem}

This remark motivates the application over lossy networks, as discussed next.

\section{Packet Drops and Sparse Connectivity} \label{sec_drop}

In this section, we focus on networks with unreliable links. This can represent various scenarios, for instance, when nodes are activated from sleep mode at a random times as is common in devices with energy harvesting capabilities \cite{Zou:2016}. In general, data transmission over wireless networks is subject to random packet dropouts which motivates us to consider a topology with random links.
Given a network topology $\mc{G}$ we model the packet drops over unreliable links at time $k$ by removing those links from the graph structure $\mc{G}(k)$. Simply speaking, given $m$ packet drops over distinct links at time $k$, $m$ links are removed from the network $\mc{G}(k)$. Then, to satisfy Assumption~\ref{ass_con} for convergence, the remaining reliable network needs to hold uniform connectivity over every $B$ time iterations for a $B \in \mathbb{Z}_{\geq 0}$. 

Assuming common knowledge, both nodes $i,j$ are aware of the delivery or loss of the messages or packets over the bi-directional link $(i,j)$ \cite{6426252,olfati_rev}. This is to (i) keep the networking balanced following Assumption~\ref{ass_con} and (ii) satisfy the all-time feasibility in the proof of Theorem~\ref{thm_conv}. For this purpose, in case the message from $i$ to $j$ is lost at iteration $k$, node $i$ \textit{does not} incorporate \textit{possibly received} message $g_l(\partial f_j (x_j(k))$ from node $j$ in its data-processing and updating state $x_i(k+1)$ via \eqref{eq_sol}. In other words, the mutual messages are either both dropped or both used in Algorithm~\ref{alg_ac}. This consideration makes the probability of packet drop different from the probability of link removal in our calculations; for packet drop rate $p_d$ (or packet delivery rate $1-p_d$) over the links from $i$ to $j$ or $j$ to $i$, the equivalent probability of link removal in our analysis follows as 
\begin{align} \label{eq_pl}
    p_l = 1-(1-p_d)^2 = 2p_d-p_d^2.
\end{align} 
implying the probability that either of the messages or both are lost.  

Recall that the notion of uniform connectivity implies that the union network over every finite number of time-steps $B$ is connected, i.e., $\mc{G}_B = \bigcup_{k}^{k+B} \mc{G}(k)$ is connected. This is much more relaxed than the all-time connectivity requirement in many literature \cite{wang2020dual,Carli_ADMM,falsone2020tracking,boyd2006optimal,rikos2021optimal,wu2021new,9683726,cherukuri2015distributed}. Intuitively, for example with $B=2$, one can assume existence of $B=2$ links between every two nodes $i,j$ over the union network $\mc{G}_B = \bigcup_{k}^{k+B} \mc{G}(k)$. Assume $p_l$ as the probability of link removal over connected graph $\mc{G}$. Recalling the definition of the union graph, the link between $i,j$ over $\mc{G}_B(k)$ is lost if the link over both $\mc{G}(k)$ and $\mc{G}(k+1)$ are removed as unreliable links. Therefore, the probability that any link is unreliable over the union network $\mc{G}_B(k)$ is $p_l^2$ for $B=2$. One can extend this to uniform connectivity over any $B$ iterations, i.e., the probability of link removal over $\mc{G}_B(k) = \bigcup_{k}^{k+B} \mc{G}(k)$ for $B \geq 1$ is $p_l^{B+1}$. This is the intuition behind the following theorem.

\begin{theorem} \label{thm_drop}
    Assume a connected network topology $\mc{G}$ with bond-percolation threshold $p_c$. There exists $B \in \mathbb{Z}_{\geq 0}$ such that $\mc{G}_{B} = \cup_{k}^{k+B}\mc{G}(k)$ remains uniformly connected with probability $1$ under any drop rate $1-\sqrt{1-p_c}<p_d<1$; this $B$ satisfies
    \begin{align} \label{eq_B*}
        (2p_d-p_d^2)^{B+1}<p_c.
    \end{align}
    For $p_d<1-\sqrt{1-p_c}$ the network $\mc{G}(k)$ remains connected at iteration $k$ with probability $1$.
\end{theorem}
\begin{proof}
   First, recall that for packet drop rate $p_d$ the probability of link removal is $p_l = 2p_d-p_d^2$ from \eqref{eq_pl}.
   For percolation threshold $p_c$, find $B^* \in \mathbb{Z}_{\geq 0}$ as the minimum value satisfying $p_l^{1+B^*}<p_c$. Note that, for  $p_l<1$, function $p_l^{1+B^*}$ is monotonically increasing on $p_l$ and decreasing on $B^*$. Then, for any $B \geq B^*$, the probability that all the links are dropped over $B^*$ time-steps is less than the percolation threshold $p_c$. This means that, for $B\geq B^*$ and $p_c<p_l<1$, the network $\mc{G}_{B}$ remains uniformly connected with probability $1$ under link removal probability $p_l = 2p_d-p_d^2$. This gives the admissible range of the packet drop rate $1-\sqrt{1-p_c}<p_d<1$.
\end{proof}

\begin{cor}
    Assume a dynamic sparse network which is not connected but uniformly-connected over $B_0 >0$ time iterations. Given the bond-percolation threshold $p_c$ associated with $\mc{G}_{B_0}$ and drop rate  $1-\sqrt{1-p_c}<p_d<1$, one can find $B^* \in \mathbb{Z}_{\geq 1}$ such that $p_l^{1+B^*} < p_c$ with $p_l = 2p_d-p_d^2$. Then, the union network $\mc{G}_B$ under the same packet drop rate remains uniformly connected for any $B \geq B_0 B^*$. 
\end{cor}

\begin{rem}
   In the case of heterogeneous and time-varying drop rates $p_{ij}(k)$ at different links $(i,j)$, one can consider $p_d = \max\{p_{ij}(k)\}$ in Theorem~\ref{thm_drop} as a conservative solution to find minimum $B^*$ value. 
\end{rem}

\section{Simulation} \label{sec_sim}

For the simulation we consider a network $\mc{G}$ of $n=20$ nodes based on the ER model shown in Fig.~\ref{fig_graphs}(TopLeft). The linking probability between every two nodes is considered as $p=0.3$. The ER theory states that for $p>\frac{1}{n-1}$, the network is connected \cite{erdos1960evolution}, i.e., the connectivity transition occurs at $np=1$. We consider different link removal probabilities $p_l = [0.64, 0.71, 0.79, 0.86, 0.93]$ over $40$-steps switching periods. From \eqref{eq_pl} the associated packet drop rates are $p_d = [0.4,0.46,0.54,0.62,0.73]$. We use MATLAB \textsf{randperm} function to randomly assign these rates over successive periods of every $200$ steps. 
Sample network topologies (after removal of unreliable links) associated with the given probabilities $p_l$ are shown in Fig.~\ref{fig_graphs}. 
For the ER graph in Fig.~\ref{fig_graphs}(a) (with $p_l=0$), we have  $\mean{\mc{N}} = 5.6$ (as average node degree) and, from Table~\ref{tab_perc}, we have $p_c = 0.177$. Note that all considered $p_l$ values are over this threshold. From Theorem~\ref{thm_drop}, we have respective values $B^* = [3,5,7,11,23]$, which implies uniform-connectivity over $B \geq 23$ steps and, thus, any sufficiently small $\eta$ satisfying Theorem~\ref{thm_conv} ensures convergence of the allocation algorithm, as described next.   

\begin{figure*}[t]
\centering
  \subfigure[$p_l=0$]
 {       \includegraphics[width=2in]{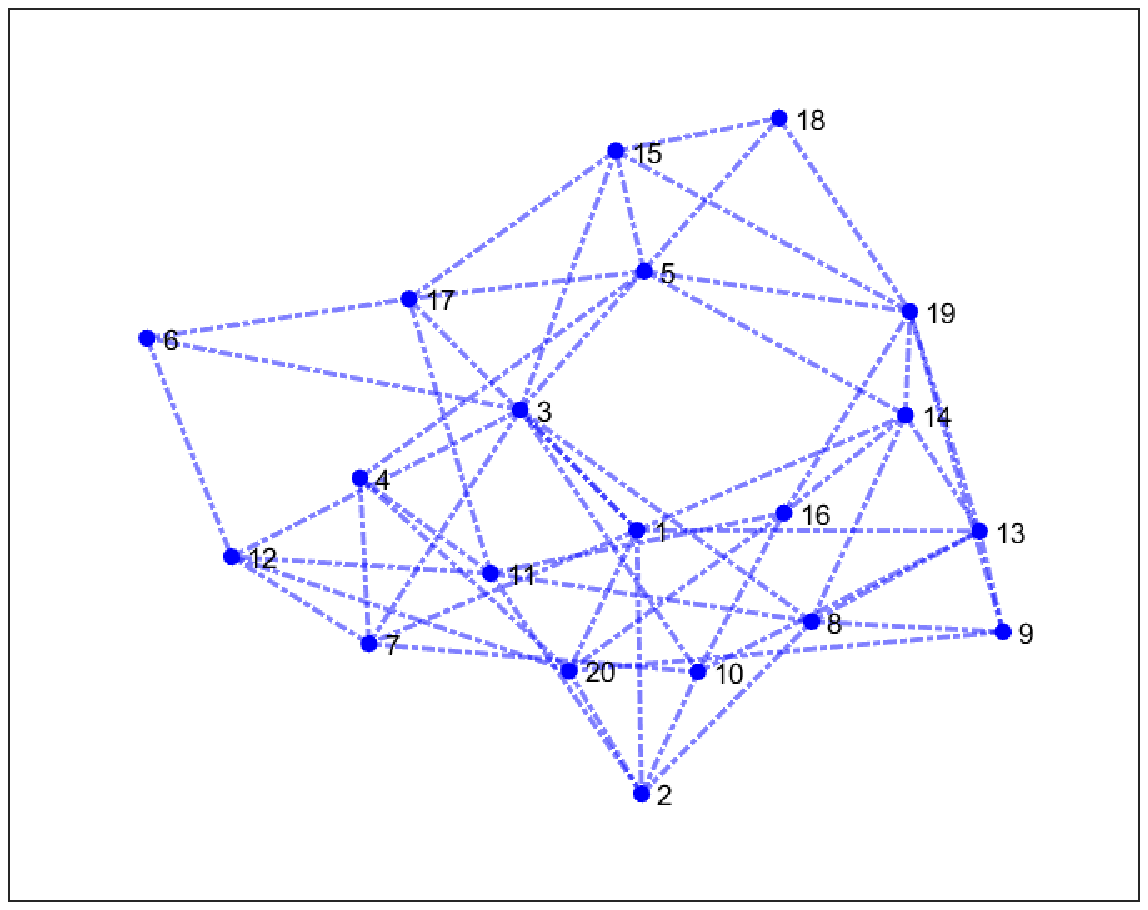}}
  \subfigure[$p_l=0.64$]
 { 		\includegraphics[width=2in]{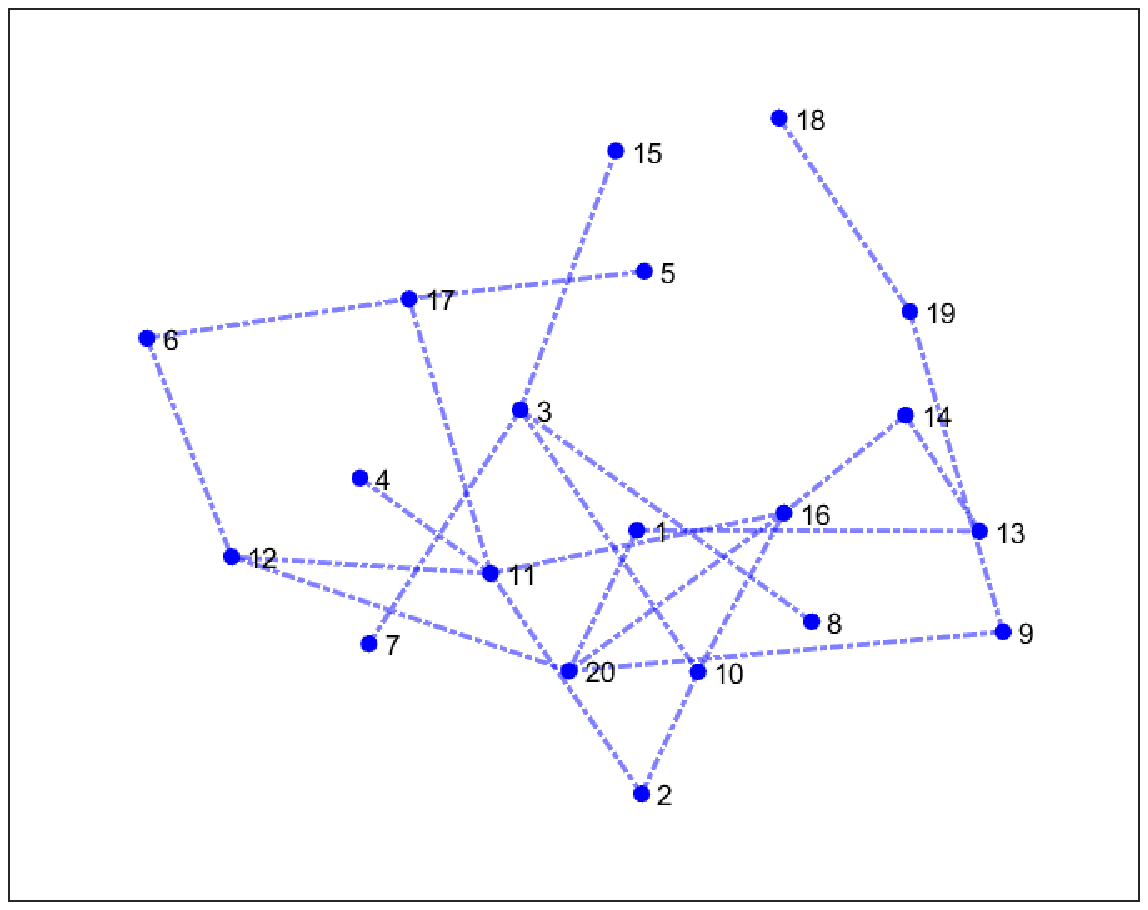}}
  \subfigure[$p_l=0.71$]
 {  		\includegraphics[width=2in]{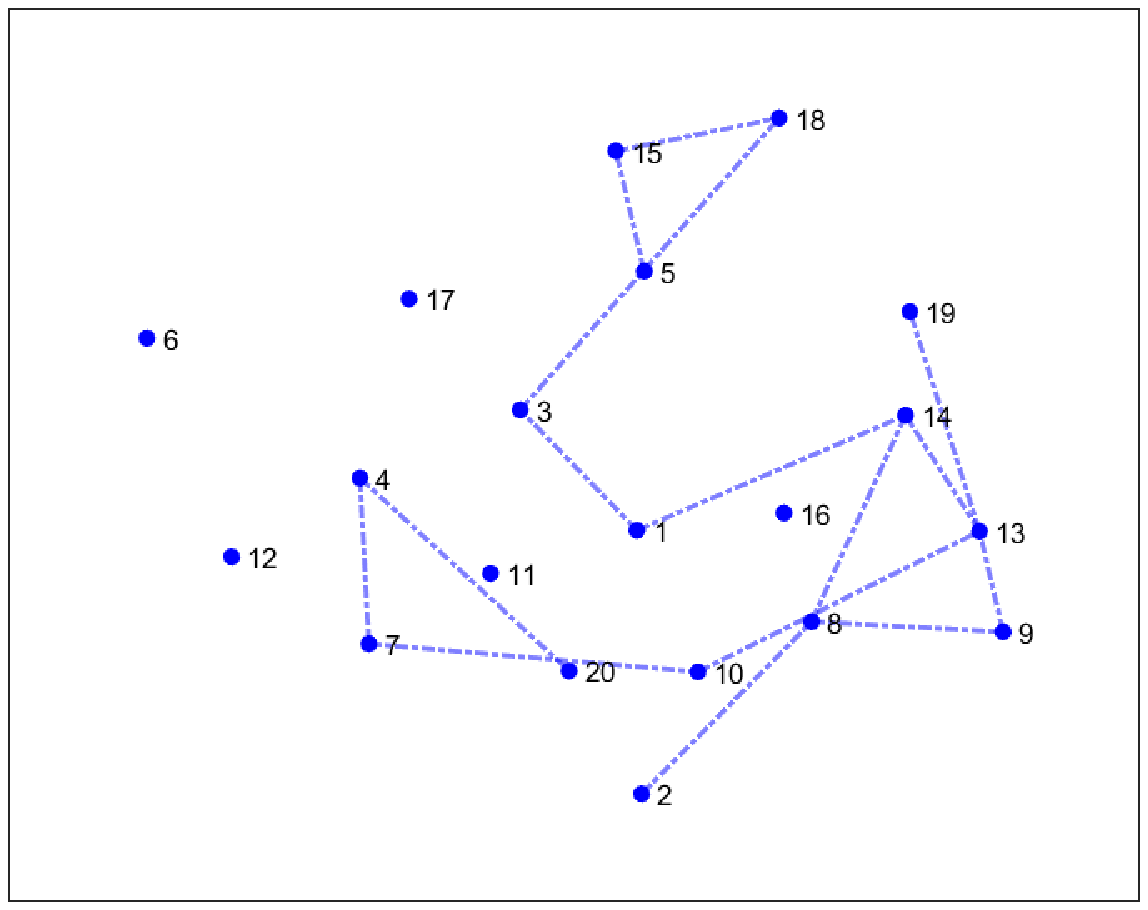}}
   \subfigure[$p_l=0.79$]
 {  		\includegraphics[width=2in]{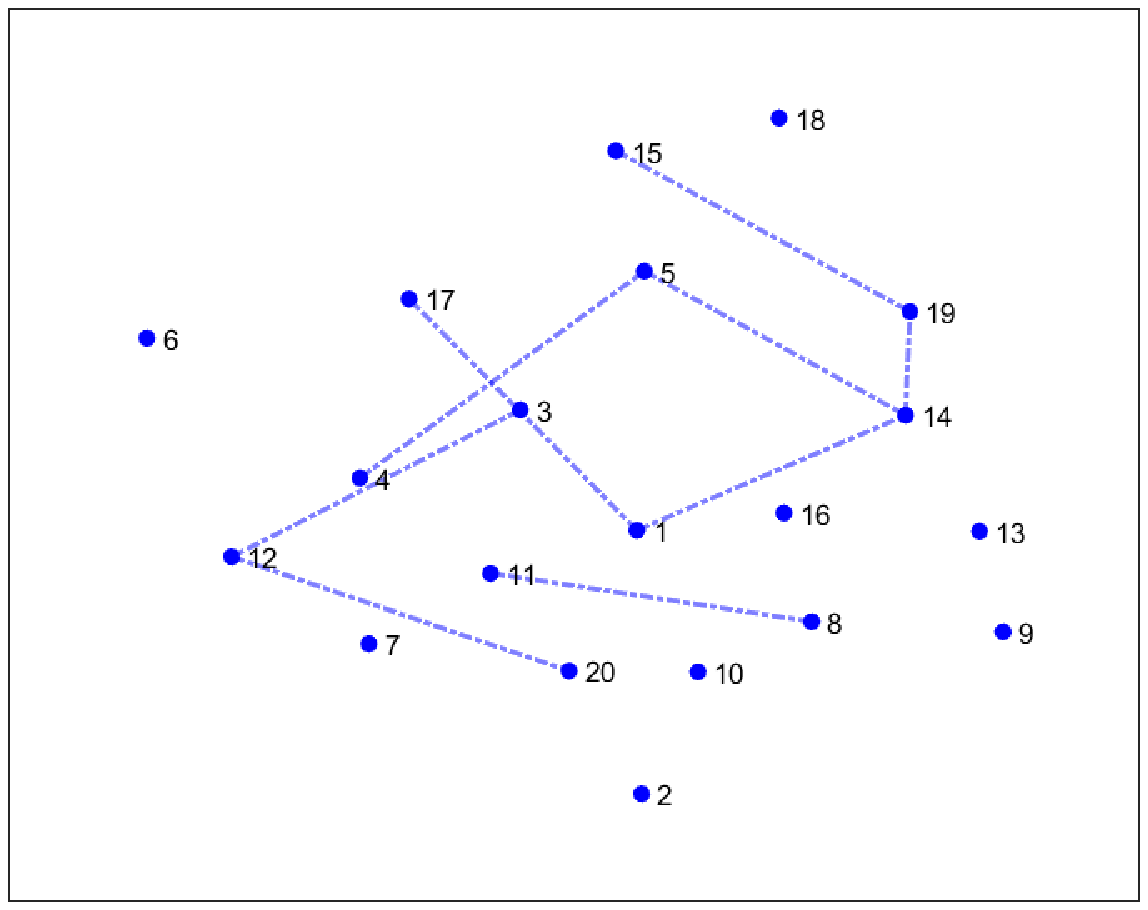}}
   \subfigure[$p_l=0.86$]
 { 		\includegraphics[width=2in]{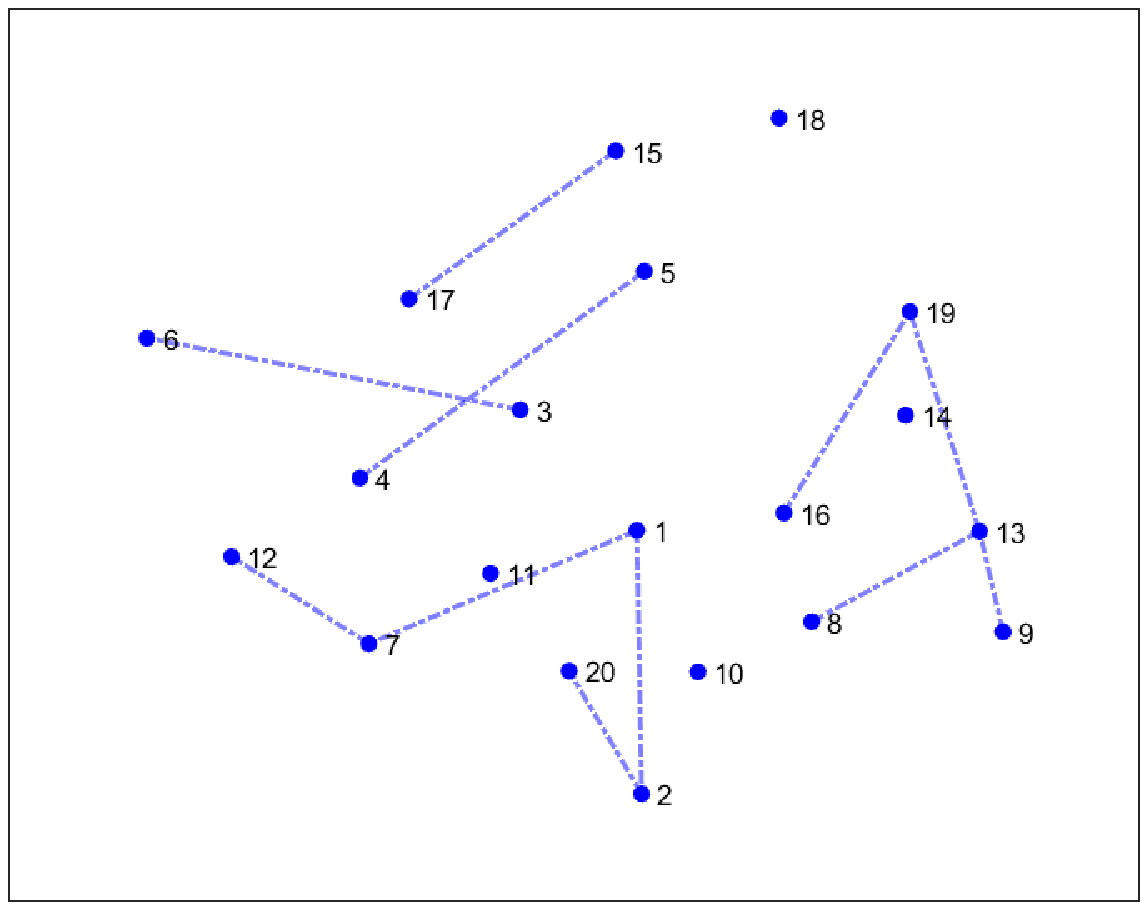}}
   \subfigure[$p_l=0.93$]
 { 		\includegraphics[width=2in]{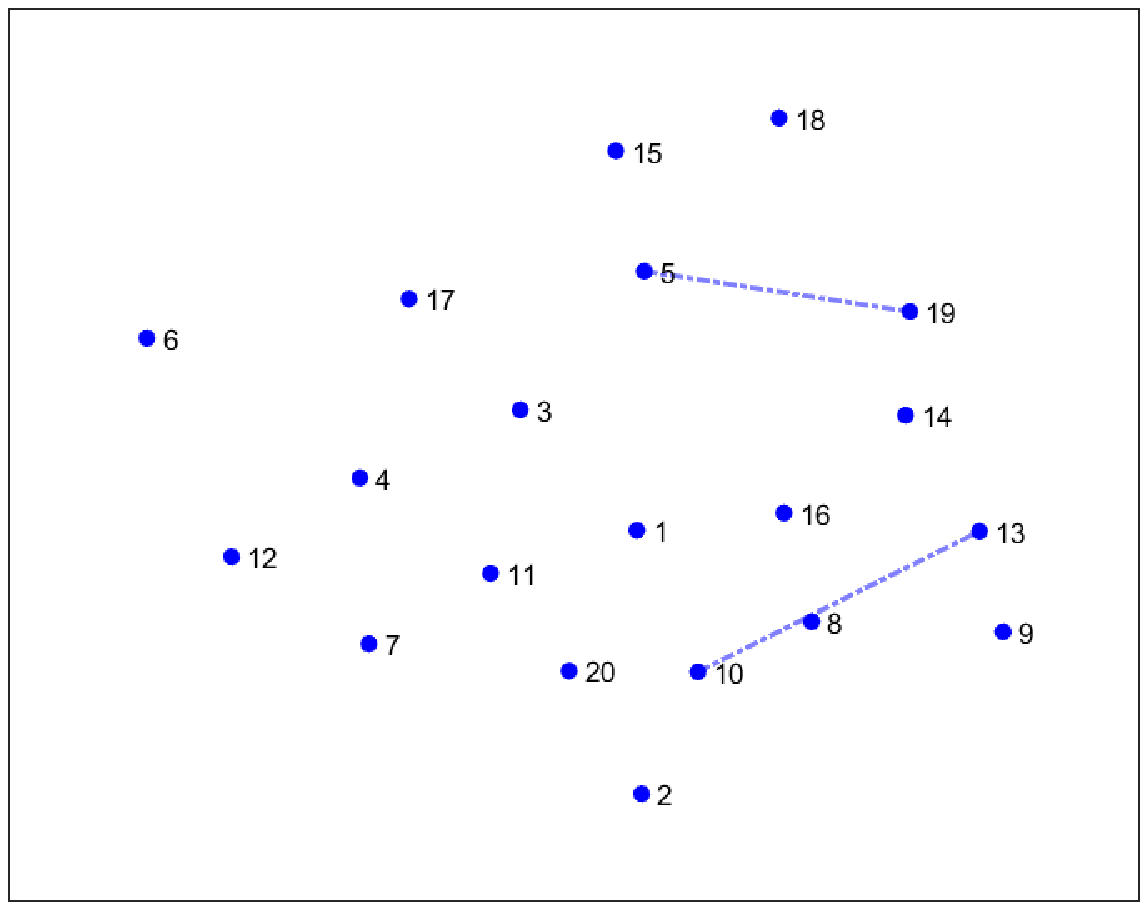}}
		\caption{This figure shows (a) the primary network and (b)-(f) the outcomes after different link removal rates $p_l$. 
		}
		\label{fig_graphs}
		\vspace{-0.5cm}
\end{figure*}

To verify, we apply Algorithm~\ref{alg_ac} over this dynamic network to solve the following allocation problem:
\begin{align} \nonumber
\sum_{i=1}^n f_i(\mb{x}_i) &=  \sum_{i=1}^n \frac{a_i}{2}(x_i-c_i)^2 + \log \left(1+\exp{(l_i(x_i-d_i))}\right) \\ \label{eq_f_quad2} &~\mbox{s.t.}~\sum_{i=1}^n x_i = b = 100 
\end{align}
with random parameters $a_i,l_i,c_i,d_i$ and  adding penalty terms $ \max \{x_i - M_i,0\}^2 + \max \{m_i - x_i,0\}^2$ to address the box constraints $ m_i = 2$, $M_i =7$. 
In the proposed dynamics \eqref{eq_sol}, we consider two example strongly sign-preserving nonlinear functions $g_n(z) = z+z^3 $ and logarithmic quantizer $g_l(z)$ with $\rho = \frac{1}{256}$. We consider random symmetric link weights $W_{ij}$ in the range $(0,10]$ (non-stochastic). For the states bounded by the given box constraints, we have $\kappa_n = 1$ and $\mc{K}_n = 147$ and $u=0.05$. For the logarithmic quantization we have $\kappa_l = 1-\frac{\rho}{2} = 0.998$ and $\mc{K}_l = 1 + \frac{\rho}{2} = 1.002$. For the given network $\mc{G}_B$ we have $\frac{\lambda_2}{\lambda_n^2} = 0.019$. This gives $\overline{\eta} = 0.0025$ as the  bound on the step rate for convergence. This gives a sufficient bound to ensure convergence and, for faster decay rates, in this simulation we choose $\eta = 0.05$. Fig.~\ref{fig_x} shows the evolution of the residual cost $\overline{F}$ and states $x_i$ in this example. 

\begin{figure} [t]
		\centering
		\includegraphics[width=3.5in]{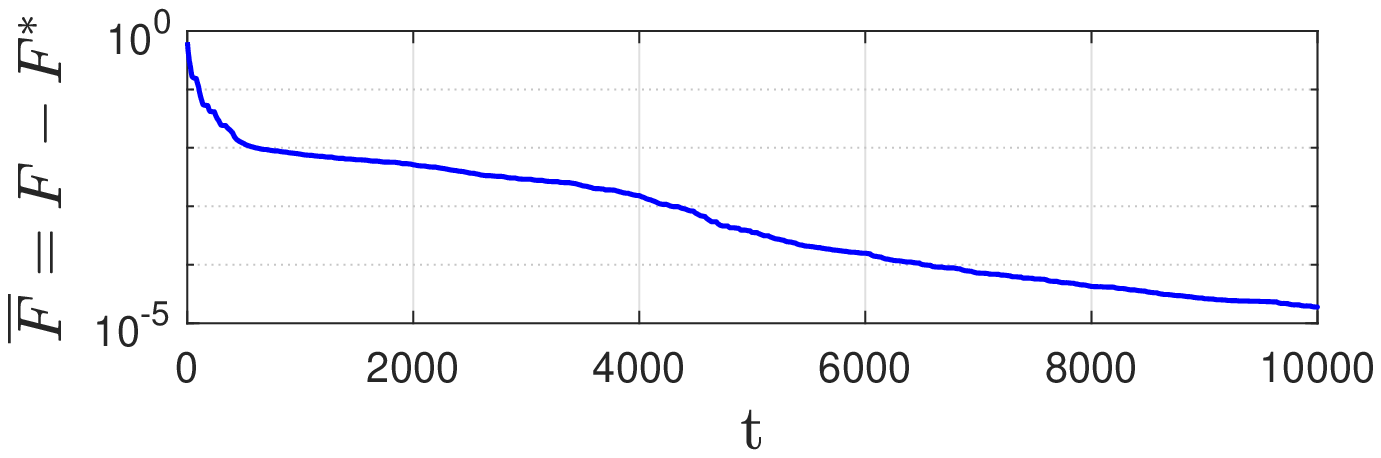}
 		\includegraphics[width=3.5in]{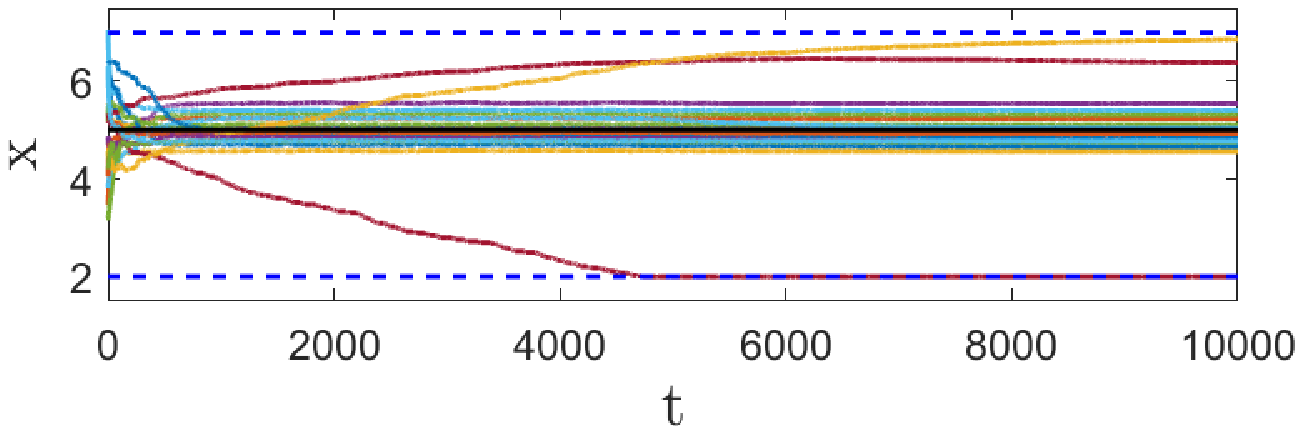}
		\caption{This figure shows (Top) the decrease in the residual cost $\overline{F}$  and (Bottom) the evolution of the assigned states over time with the dashed blue lines representing the box constraints and the solid black line as the average of states (to verify the all-time feasibility condition). }
		\label{fig_x}
		\vspace{-0.75cm}
\end{figure}

\section{Discussions and Concluding Remarks} \label{sec_con}
This paper provides a robust approach for distributed resource allocation over networks with unreliable links that resemble communication channels subject to packet drops. By relaxing (i) the all-time network connectivity requirement to uniform connectivity and (ii) the dynamic weight-stochastic condition to weight-balanced links, the proposed allocation algorithm is proved to converge under different packet drop rates, but over a longer time horizon. This approach finds application in mobile multi-agent networks and wireless communication networks with the inevitable high rate of packet losses. The proposed algorithm has many other advantages over the existing allocation solutions. For example, it can address nonlinearities on the agents' dynamics. This nonlinear model may resemble (i) physics of the system and inherent constraints in application, e.g., quantized information exchange \cite{mrd_2020fast,rikos2021optimal}, clipping, or actuator saturation, and (ii) purposely designed dynamics to  suppress impulsive noise or improve the convergence time \cite{mrd20211st}. The solution is not limited to the quadratic cost models, e.g., in CPU scheduling and economic dispatch, but it can handle general non-quadratic cost models, e.g., due to additive barrier functions \cite{wu2021new} and penalty terms \cite{nesterov1998introductory,bertsekas1975necessary}. As a future research direction, we are extending these results to address possible time-delays in data exchange over the links.

Another interesting preventive approach is \textit{survivable network design} via  \textit{$Q$-edge-connected graphs} which is  robust to link removal. This problem aims to design the network such that it remains connected after removing any subset of size (up to) $Q$ links, or it preserves a prescribed routing criterion (up to) a certain disruption cost \cite{6004644,panigrahi2011survivable}. By such a design, dropping up to $Q$ packets over the network, a route/path of package delivery between every two nodes is guaranteed over time, ensuring convergence of our allocation algorithm. Recall that existing survivable design algorithms ensure connectivity at every time-instant and extension to more relaxed uniform connectivity over time (as described in Assumption~\ref{ass_con}) is another promising future research direction. 

\bibliographystyle{IEEEbib}
\bibliography{bibliography}
\end{document}